\documentclass[letterpaper, 10 pt, conference]{ieeeconf}  

\IEEEoverridecommandlockouts

\usepackage{cite}
\usepackage{amsmath,amssymb,amsfonts}
\usepackage{algorithm,algorithmic}
\usepackage{amsthm}
\usepackage{ upgreek }
\usepackage{algorithmic}
\usepackage{graphicx}
\usepackage[font=scriptsize]{subcaption}
\usepackage{textcomp}
\usepackage{xcolor}
\usepackage{multicol,lipsum}
\usepackage{xcolor}
\theoremstyle{plain}
\newtheorem{thm}{Theorem}
\newtheorem{lem}{Lemma}
\newtheorem{prop}{Proposition}

\theoremstyle{definition}

\newtheorem{rem}{Remark}
\newtheorem{prob}{Problem}
\usepackage{datetime}
\usepackage{mathtools}

\title{\bf
Construction of control barrier function and $C^2$ reference trajectory for constrained attitude maneuvers  }

\author{Xiao Tan, and  Dimos V. Dimarogonas 
\thanks{*This work was supported by the H2020 ERC Starting Grant BUCOPHSYS, the SSF COIN project, the Swedish Research Council (VR) and the Knut och Alice Wallenberg Foundation.}
\thanks{Xiao Tan and Dimos V. Dimarogonas are with the Division of Decision and Control Systems, KTH Royal Institute of Technology, Stockholm, Sweden.
        {\tt\small xiaotan@kth.se, dimos@kth.se}}%
}

\begin{document}

\maketitle
\thispagestyle{plain}
\pagestyle{plain}

\begin{abstract}
 Constrained attitude maneuvers have numerous applications in robotics and aerospace. 
 In our previous work, a general framework to this problem was proposed with resolution completeness guarantee. However, a smooth reference trajectory and a low-level safety-critical controller were lacking.  In this work, we propose a novel construction of a $C^2$ continuous reference trajectory based on B\'ezier curves on $ SO(3) $ that  evolves within predetermined cells and eliminates previous stop-and-go behavior. Moreover, we propose a novel zeroing control barrier function on $ SO(3) $ that provides a safety certificate  over a set of overlapping cells on $ SO(3) $ while avoiding nonsmooth analysis. The safety certificate is given as a linear constraint on the control input and implemented in real-time.  A remedy is proposed to handle the states where the coefficient of the control input in the linear constraint vanishes.  Numerical simulations are given to verify the advantages of the proposed method.


%
%
%

\end{abstract}

\section{Introduction}

The study of the attitude (orientation) control problem arised from early space and aerial applications and became prevalent in autonomous robotic systems.  A recent trend in this field is to study this problem using Lie group theory \cite{koditschek1988application,bayadi2014almost,lee2012robust,berkane2017hybrid}, motivated by the fact that there exists no  attitude parametrization other than the special orthogonal group $ SO(3) $ that  both globally and uniquely represents the rotational space and avoids singularities and the unwinding phenomenon.    Many safe-critical applications, such as space telescopes observing some celestial regions while avoiding bright stars \cite{hablani1999attitude}, and the anisotropy sensitive imaging and communication payloads on UAVs and AUVs,  motivate further study of the attitude planning and control problem in the presence of orientation constraints (i.e., unfeasible rotational regions).  

There exist two main approaches for the constrained attitude maneuver problem: the potential-function \cite{lee2014feedback,kulumani2017constrained,hu2018anti} and the planning based methods \cite{frazzoli2001randomized,kjellberg2015discretized,biggs2016geometric}. By delicately designing a potential function, the feedback controller utilizes the negative gradient to guide the rotational movement. Generally speaking, potential-function based methods are easy to implement, but the state trajectory  may get stuck at local minima (where the gradient vanishes) and requires convexity of the safe regions. On the contrary, planning-based methods try to find a feasible trajectory leading to the target state and then a tracking controller is utilized. This approach, however, mainly suffers from the lack of completeness guarantees, i.e., derive a solution if it exists,  and safety guarantees, i.e., a certificate that the actual trajectory will not derivate from the reference and remain in the safe region.  

In our previous work \cite{tan2020constrained}, a hierarchy framework was proposed for the constrained attitude maneuver problem consisting of 1) discretizing the rotation group $ SO(3) $ into finite overlapping cells, 2) planning over the cells, and 3) reference trajectory generation and tracking control. Viewing the sampling step as the resolution level, we guarantee a feasible path is to be found in finite time when one exists at that resolution. However, the reference trajectory in \cite{tan2020constrained} is constructed by the concatenation of geodesic paths  and has to reach zero velocities at end points for each sub-maneuvers. This is a potential drawback as it requires the vehicle to stop and go from configuration to configuration. Moreover, no safety guarantee is developed for the low-level tracking controller.  

In this work, we follow the framework in \cite{tan2020constrained} and further construct a $ C^2 $ reference trajectory and develop a safety certificate by designing zeroing control barrier functions on $ SO(3) $. The $ C^2  $ reference trajectory is generated by a  B\'ezier curve on $ SO(3) $. By choosing the controlling points carefully, we show that the constructed curve is of $ C^2 $ continuity, connects the initial and target orientations, and evolves within the set of given cells. 

This paper has two additional contributions: 1) Noting that the safety region is a union of a set of overlapping cells, we formulate a smooth control barrier function and thus avoid the nonsmooth analysis as the case in \cite{glotfelter2017nonsmooth}. The proposed formulation is at the cost of shrinking the safety region and this conservativeness can be explicitly adjusted by a user-defined parameter; 2) Since the Lie derivative of the control barrier function candidate vanishes at certain states, existing high-order control barrier function design methods \cite{wences2020correct,xiao2019control} are not directly applicable. To address this issue, we introduce a remedy with a proof to render the constraint on the control input feasible for all states in the safety region. All results are illustrated though relevant simulations.

\section{Preliminaries and problem formulation}

The set of real, non-negative real, and  positive integer numbers are denoted as $ \mathbb{R}, \mathbb{R}_{\ge}, \mathbb{N}  $, respectively.  $ \mathbb{R}^n $ denotes the $ n $-dimensional Euclidean space. The 2-norm of a vector $ {x} \in \mathbb{R}^n $ is $ \| {x} \|_2 := \sqrt{{x}^\top {x}}  $. $ I $ is the $ 3 $-dimensional identity matrix. The Frobenius norm of $ A $ is defined as $ \Vert A \Vert_F = \mathbf{ tr}(A^\top A) $, where  $ \mathbf{ tr}(\cdot)  $ denotes the trace of a matrix. The Lie derivatives of a function $h(x)$ for the system $\dot{x} = f( x) +  g( x)  u$ are denoted by $L_{f} h:=\frac{\partial h}{ \partial x} f(x) $ and $L_{g} h:=\frac{\partial h}{ \partial x} g(x)$, respectively.  A continuous function $ \alpha: (-b,a) \to (-\infty, \infty)  $ is said to belong to \textit{extended class} $ \mathcal{K} $ for some $ a, b > 0 $ if it is strictly increasing and $ \alpha(0) = 0 $ \cite{ames2016control}.


Any rotation matrix is an element of the Special Orthogonal group $ SO(3) := \{ R \in \mathbb{R} ^{3\times 3} \vert R^\top R=RR^\top= I, \det{(R)} = 1 \} $ which, when associated with the matrix multiplication operation, forms a Lie group. The associated Lie algebra, denoted by $ \mathfrak{so}(3) $, consists of the set of all skew-symmetric $ 3\times 3 $ matrices, {\it i.e.,} $\mathfrak{so}(3):=\{\Omega\in\mathbb{R}^{3\times 3} : \Omega^\top=-\Omega\}$. The Lie bracket for $  \mathfrak{so}(3) $ is given as $ [V,W] = VW - WV $, for any $ V,W \in \mathfrak{so}(3)$. The map $ [(\cdot)]_{\times} : \mathbb{R}^{3} \rightarrow \mathfrak{so}(3) $ and its inverse map $ \vee:\mathfrak{so}(3) \rightarrow \mathbb{R}^{3} $  are explicitly defined as
$ x = \begin{psmallmatrix}
x_1 \\ x_2\\  x_3
\end{psmallmatrix}  \xrightleftharpoons[(\cdot)^{\vee}]{[(\cdot)]_{\times}} [x]_{\times} = \begin{psmallmatrix}
0   & -x_3 & x_2\\
x_3 &  0   & -x_1\\
-x_2& x_1 & 0
\end{psmallmatrix}.  $ The Lie algebra $\mathfrak{so}(3)$ allows to represent rotation matrices on $SO(3)$ via  the matrix exponential $  \exp(\cdot)$. For $[x]_{\times}\in \mathfrak{so}(3) $,
$\exp([x]_{\times}) = I+\frac{\sin(\|x\|_2)}{\|x\|_2}[x]_{\times}+\frac{1-\cos(\|x\|_2)}{\|x\|_2^2}[x]_{\times}^2 $ when $x\neq 0$, and $\exp([x]_{\times}) = I $ otherwise \cite{bullo2004geometric}. For all rotation matrices $R$ with $ \mathbf{ tr}( R) \neq -1$, the exponential map admits an inverse logarithmic map given by $\log(R) = \frac{\theta(R)}{2\sin(\theta(R))}(R-R^\top)$ when $R\neq I$, and $\log(R) = 0$ otherwise, where $\theta(R):=\arccos\left((\mathbf{ tr}(R)-1)/2\right)$ is the rotation angle associated to $R$ \cite{bullo2004geometric}.

A \textit{path} $ F(\cdot) $ in $ \mathcal{A} $  connecting $ R_{1} \in SO(3)$ and $ R_{2} \in SO(3) $, where $\mathcal{A}$ is a subset of $SO(3)$, is defined as a continuous function $ F: [a,b] \rightarrow \mathcal{A}$ with $ F(a) = R_{1} $ and $ F(b) = R_{2} $. If there exists such a path  $ F(\cdot) $, we say $ (R_{1}, R_{2}) $ is \textit{connected}. If any two points in $ \mathcal{A} $ are connected, then we call the set $ \mathcal{A} $ \textit{path-connected}.
Given  any $ R_1, R_2 \in SO(3) $ with $\mathbf{tr}(R_1^\top R_2)\neq -1$, the {\it geodesic path} between $ R_1 $ and $ R_2 $ is $ R(\tau) = R_1 \exp(\tau \log(R_1^\top R_2)), 0\le \tau \le 1 $. The {\it angular distance} between $ R_1, R_2 $ is given by $\textup{d}(R_1, R_2) := \lVert \log{(R_1R_2^\top)} \rVert_2$. 



  In \cite{tan2020constrained}, we proposed a general framework for the constrained attitude maneuver problem consisting of  $SO(3)$ space partitioning, planning, and reference trajectory generation. We briefly recall the results here. Let the sampling set $ U : = \{R_1, R_2, \dots, R_i, \dots, R_n\} $ be a finite set with $ n $ elements in $ SO(3) $ and let $ \mathcal{N}^{\prime} := \{ 1,2, \dots, n \} $ be an index set. For each $i\in\mathcal{N}^{\prime}$, define the cell region $ S_i $ as the open ball centered at  $ R_i $ with a radius $ \theta \in (0,\pi/2)$, i.e., $S_i := \{ R\in SO(3): \textup{d}(R, R_i)< \theta\},  \forall i \in \mathcal{N}^{\prime}.$  The neighborhood set $ N_{i}  $ of $ R_i $ is defined as $N_{i} := \{  R \in U :   \textup{d}(R, R_i) < 2\theta, R \neq R_i \}, \forall i \in \mathcal{N}^\prime$. Cells $ S_i, S_j $ are  \textit{adjacent} if $S_i\cap S_j\neq \emptyset $.

  By choosing $ U $ and $ \theta $ such that the conditions in  \cite[Theorem 1]{tan2020constrained}  are satisfied,  we guarantee that for an arbitrary cell $S_i\in \mathcal{N}^{\prime}$, it has adjacent cells; the center points of adjacent cells are strictly outside of $S_i$; the union of the cells covers the whole $SO(3)$ space. Mathematically, 
 	\begin{enumerate}
 		\item[i.]  For all $   i \in \mathcal{N}^{\prime}, N_{i} \neq \emptyset $;
 		\item[ii.] For all $   i, j \in \mathcal{N}^{\prime}, i\neq j $, we have $  R_j \notin S_i $;
 		\item[iii.] For  all $ R_i \in U $, and all $ R_j \in N_i $, $	\theta < \textup{d}(R_i, R_j) < 2 \theta$;
 		\item[iv.]  $		\underset{i\in \mathcal{N}^{\prime}}{\cup}S_i = SO(3).$
 	\end{enumerate}

\begin{lem}[\hspace*{-3px} \cite{tan2020constrained}] \label{lem:convex_cell}
	For any cell $ S_i, i\in \mathcal{N}^\prime $ and two arbitrary points $ R_{i1}, R_{i2} \in S_i $, the geodesic path between $ R_{i1} $ and $ R_{i2} $ is within $ S_i$, i.e., for any $ R_{i1}\in S_i, R_{i2} \in S_i  $, $R(\tau) =  R_{i1} \exp(\tau \log(R_{i1}^\top R_{i2})) \in S_i, 0\le \tau \le 1  $, holds.
\end{lem}

\begin{lem}[\hspace*{-6px}\cite{tan2020constrained}] \label{lem:cell_to_cell}
	The geodesic path between any two neighboring sampling rotations $ R_i $ and $ R_j $ is within $ S_i \cup S_j $, i.e.,  $R(\tau) =  R_{i} \exp(\tau \log(R_{i}^\top R_{j})) \in S_i\cup S_j, 0\le \tau \le 1  $.
\end{lem}

 We approximate a generic safe attitude zone on $ SO(3) $ by a set of cells $ \{ S_i \}, i \in \mathcal{N}, \mathcal{N} \subset \mathcal{N}^\prime  $ and a graph search algorithm is utilized that gives out a sequence of cells whenever feasible at the given resolution level. Without loss of generality, by re-labeling the cells, we assume that the initial orientation $ R_0 \in S_1$, the target orientation  $ R_f \in S_m$, $ S_i$  and $ S_{i+1} $ are adjacent cells for $ i \in \{ 1, \cdots, m-1\}$.  Based on Lemmas \ref{lem:convex_cell}$, $\ref{lem:cell_to_cell}, a center-to-center attitude maneuver was then designed, as illustrated in blue dash line in Fig. \ref{fig:trajectory_illustration}. Though the reference trajectory is guaranteed to be within the feasible region, it is not favorable in practice as it requires the rigid-body to reach zero velocities at the end points for each sub-maneuvers.  
 


\subsection{ Problem formulation}

The attitude dynamics of a rigid body are given by
\begin{equation} \label{eq:dyn}
\left\{ \begin{array}{l}
\dot{R} = R [\omega]_{\times}, \\
J \dot{\omega} + [\omega]_{\times}J\omega = u,
\end{array}\right.
\end{equation}
where the attitude $ R\in SO(3) $, $ \omega \in \mathbb{R}^3 $ is the angular velocity in the body-fixed frame, $ J $ is the constant and known inertia matrix and $ u \in \mathbb{R}^3 $ is the input torque.  Given a set of cells $\{ S_i \}, i \in \mathcal{N}, \mathcal{N} \subset \mathcal{N}^\prime $, we call a trajectory  $ \gamma: t \mapsto R(t) $ is \textit{safe} if the trajectory always evolves within $\cup_{i \in \mathcal{N}} S_i$.

The control scheme consists of two parts: reference generation and trajectory tracking.

\begin{prob}{(\bf reference generation)} \label{prob:1}
Given a set of cells $\{ S_i \}, i \in \mathcal{N}, \mathcal{N} \subset \mathcal{N}^\prime $ such that $ \cup_{i \in \mathcal{N}} S_i $ is path-connected. For any given $ R_0, R_f \in \cup_{i \in \mathcal{N}} S_i $, find a $ C^2 $ curve $ \gamma : \mathbb{R}_{\ge} \to \cup_{i \in \mathcal{N}} S_i $  such that $ d\gamma /dt(0) = d\gamma /dt(T) =0,  D^2\gamma /dt^2 (0) = D^2\gamma /dt^2(T) = 0 $\footnote{  For a curve $\gamma: \mathbb{R}\to SO(3)$, $D^2\gamma /dt^2 $ represents the geometric acceleration instead of the second-order total derivatives, following the terminology in \cite{bullo2004geometric}.}, $  \gamma (0) = R_0, \gamma (t) = R_f, \forall t \geq T$.
\end{prob}

\begin{prob}{(\bf trajectory tracking)} \label{prob:2}
Given a $ C^2 $ curve $ \gamma:\mathbb{R}_{\ge} \to \cup_{i \in \mathcal{N}} S_i $, design a control law  $ u $ for the system \eqref{eq:dyn} such that $ R(t) \in \cup_{i\in \mathcal{N}} S_i $ for $ t\ge 0  $ and  $ \lim_{t \to \infty}  R(t) =  \gamma (t) $.
\end{prob}

In the following, we will solve Problem \ref{prob:1} and Problem \ref{prob:2} in Section \ref{sec:bezier_curve} and Section \ref{sec:control_barrier}, respectively.


\section{ B\'ezier curve construction over cells } \label{sec:bezier_curve}
In this section, we construct a reference trajectory based on B\'ezier curve on $ SO(3) $ that solves Problem \ref{prob:1}.   B\'ezier curve is chosen here  because   De Casteljau Algorithm, which generates B\'ezier curves, is in essence a geometric construction, and naturally generalizes to $SO(3)$ manifold, while other splines are not defined / easy to compute on $SO(3)$. 

\subsection{De Casteljau Algorithm on $ SO(3) $}
  We briefly recall De Casteljau Algorithm  from \cite{park1995bezier} as follows.   Taking the geodesics on $ SO(3) $ as the analog of straight lines,  De Casteljau Algorithm connects two points in $SO(3)$ via an iterative linear interpolation process. Let $ n+1 $ ordered points of $ SO(3) $ be $  \{ x_0,x_1,\cdots,x_n \} $. The sequence of curves is defined recursively on $ SO(3) $ as
 \begin{multline} \label{eq:bezier_curve_recursive}
 	x_i^k(\tau) = x_{i-1}^{k-1}(\tau)  \exp(\tau \log([x_{i-1}^{k-1}(\tau)]^\top x_i^{k-1}(\tau))), \\ k = 0,1,\cdots,n, \quad  i = k,k+1, \cdots,n,
 \end{multline} 
where $ x_i^0(\tau) = x_i $. The B\'ezier curve is then given by
\begin{equation}\label{eq:bezier_curve}
x_n^n(\tau) =  x_{n-1}^{n-1}(\tau)  \exp(\tau \log([x_{n-1}^{n-1}(\tau)]^\top x_n^{n-1}(\tau))).
\end{equation}

\begin{lem}[\hspace*{-3px}\cite{crouch1999casteljau}] \label{lem:De_Cas}
	Let $ n+1 $ ordered points of $ SO(3)$ be $  \{ x_0,x_1,\cdots,x_n \} $. The corresponding B\'ezier curve generated from \eqref{eq:bezier_curve} satisfies the following boundary conditions:
	\begin{equation} \label{eq:bezier_boundary_condition}
		\begin{gathered}
		x_n^n(0) = x_0, x_n^n(1) = x_n, \\
		\frac{d}{d\tau}x_n^n(\tau)\vert_{\tau = 0} = n x_0 V_0, 		\frac{d}{d\tau}x_n^n(\tau)\vert_{\tau = 1} = n x_n V_{n-1},\\
		\frac{D^2}{d\tau^2}x_n^n(\tau)\vert_{\tau = 0} = n(n-1) x_0 \Upupsilon^{-1}_0(V_1 - V_{0}), 	\\
		\frac{D^2}{d\tau^2}x_n^n(\tau)\vert_{\tau = 1} = n(n-1) x_n \Upupsilon^{-1}_1(V_{n-1} - V_{n-2}),\\
		\end{gathered}
	\end{equation}
	where $ V_{i} = \log(x_{i}^\top x_{i+1}) \in \mathfrak{so}(3), i = 0,1,\cdots,n-1$, $ \Upupsilon^{-1}_0  $ and $ \Upupsilon^{-1}_1 $ are respectively the inverses of the operators $  \Upupsilon_0(W)  = \int_{0}^{1} \exp(u \textup{ad} V_0)W du,  \Upupsilon_1(W) = \int_{0}^{1} \exp(- u \textup{ad} V_{n-1}) W du. $
\end{lem}

 For any $ W \in \mathfrak{so}(3) $, the operator $ \Upupsilon_0: \mathfrak{so}(3) \to \mathfrak{so}(3) $ is given explicitly by the power series $\int_{0}^{1} \exp(u \textup{ad} V_0) W du = \int_{0}^{1} W  + u[V_0,W] + \frac{u^2}{2!}[V_0,[V_0,W]] + \cdots du$.
The operator $ \Upupsilon_1(W)$ is given in a similar way. It can be easily verified that both operators $ \Upupsilon_0,  \Upupsilon_1 $ are linear transformations on $ \mathfrak{so}(3) $, i.e., $ \Upupsilon_i(aW) = a \Upupsilon_i(W), \Upupsilon_i(W + V) =  \Upupsilon_i(W)+ \Upupsilon_i(V)$ for $ W,V \in \mathfrak{so}(3), a \in \mathbb{R}, i= 1,2 $. In \cite{crouch1999casteljau}, it is shown that the inverse operator $ \Upupsilon_i^{-1} $  always exists for $i=1, 2$. 

Lemma \ref{lem:De_Cas} introduces the analytical expression of the velocity and geometric acceleration at the boundary point that will facilitate our construction of the reference trajectory with $ C^2 $ continuity.

\begin{rem}
Note that the De Casteljau algorithm in \eqref{eq:bezier_curve_recursive} is not well-defined for arbitrary points $ x_0, x_1, \cdots, x_n $ on $ SO(3) $ when $ \mathbf{ tr}([x_{i-1}^{k-1}(\tau)]^\top x_i^{k-1}(\tau)) =  -1 $ occurs.

\end{rem}

\subsection{B\'ezier curve construction in one cell}
In this subsection, we demonstrate  the procedure to design the controlling points and the properties of the constructed B\'ezier curve. 

Given a cell $ S_i, i \in \mathcal{N}$ with center point $ x_2 $ and two arbitrary points $ x_0,x_4 \in S_i $, the curve $ c_{x_0,x_2,x_4}:[0,1] \to SO(3) $ is generated as follows:  first add controlling points $x_1, x_3$ as the midpoints of $x_0,x_2$, and $x_2,x_4$, respectively; then applying De Casteljau algorithm with $n = 4$. The construction is given in Algorithm \ref{alg:1}:
\begin{algorithm} [h!]
	\caption{B\'ezier curve construction in one cell.} \label{alg:1}
	\begin{algorithmic}[1]
		\renewcommand{\algorithmicrequire}{\textbf{Input:}}
		\renewcommand{\algorithmicensure}{\textbf{Output:}}
		\REQUIRE start point $ x_0 $, cell center $ x_2 $, end point $ x_4 $
		\ENSURE  curve $ c_{x_0,x_2,x_4} $
		\STATE $  x_1 \leftarrow x_0 \exp(1/2 \log(x_{0}^\top x_{2}))$
		\STATE $  x_3 \leftarrow  x_2 \exp(1/2 \log(x_{2}^\top x_{4}))$
		\STATE calculate a sequence of curves recursively as in \eqref{eq:bezier_curve_recursive} given the ordered points $ \{x_0,x_1,x_2,x_3,x_4\} $ with $ n = 4 $ 
		\RETURN $  c_{x_0,x_2,x_4} \leftarrow x_4^4$
	\end{algorithmic} 
\end{algorithm}

 Noticing that $ V_0 =  \log(x_0^\top x_1) = 1/2\log(x_0^\top x_2), V_1 = \log(x_1^\top x_2) =  1/2\log(x_0^\top x_2) $, we have $ V_0 = V_1 $. Similarly, $ V_2 = V_3 $. From Lemma 3, we can easily check that $ c_{x_0,x_2,x_4}(0) = x_0, c_{x_0,x_2,x_4}(1) = x_4 $. The velocities at the boundary point are
\begin{equation} \label{eq:BevierVel}
\begin{gathered}
 \tfrac{d}{d\tau} c_{x_0,x_2,x_4} (0) = 2x_0\log(x_0^\top x_2), \\
  \tfrac{d}{d\tau} c_{x_0,x_2,x_4} (1) = 2x_4\log(x_2^\top x_4), \\
\end{gathered}
\end{equation}
and the geometric accelerations are given by 
\begin{equation} \label{eq:BezierAcc}
\begin{gathered}
 \tfrac{D^2}{d\tau^2} c_{x_0,x_2,x_4} (0) = 12x_0  \Upupsilon^{-1}_0(V_1 - V_0 ) = 0, \\
   \tfrac{D^2}{d\tau^2} c_{x_0,x_2,x_4} (1) = 12x_n  \Upupsilon^{-1}_1(V_{3} - V_{2} ) = 0, \\
\end{gathered}
\end{equation}
noticing that $ V_1 - V_0 = V_3 - V_2 = 0 $ and $ \Upupsilon_0^{-1}, \Upupsilon_1^{-1} $ being linear transformations.

In addition to these  explicitly expressed velocities and geometric accelerations at the endpoints, we have another nice property of the constructed curve $ c_{x_0,x_2,x_4} $.

\begin{prop} \label{prop:CurveInCell}
Given arbitrary $ n+1 $ ordered points $ \{x_0, x_1, x_2,$  $\cdots,x_n \}$ such that $ x_i \in S$, $ i = 0,1,\cdots,n $, where $ S $ is a ball region in $ SO(3) $ with radius $ \theta \in (0,\pi/2) $. The B\'ezier curve $ x_n^n(\tau) $ generated from \eqref{eq:bezier_curve} always exists and evolves in $ S $, i.e., $ x_n^n(\tau) \in S, 0 \le \tau \le 1. $
\end{prop}

\begin{proof}
This can be shown by induction on $ k $. For $k = 0$, as $ x_i^{0} (\tau) = x_i $, we have $ x_i^{0} (\tau) \in S $ for $ i =  0,1,\cdots,n$ and $ \textup{d}(x_{i-1}^{0} (\tau) , x_i^{0} (\tau) ) < \pi $, which means $ \log([x_{i-1}^{0}(\tau)]^\top x_i^{0}(\tau)) $ is well-defined and from Lemma \ref{lem:convex_cell}, $ x_i^1(\tau) \in S $. 
 For any $ k \in \{ 1,2,\cdots, n\} $, assume that $ x_{i}^{k-1}(\tau) \in S $ for $ i = k-1,\cdots,n  $, then $ x_i^k(\tau) = x_{i-1}^{k-1}(\tau)  \exp(\tau \log([x_{i-1}^{k-1}(\tau)]^\top x_i^{k-1}(\tau))) $
is well-defined as $ \text{d}( x_{i-1}^{k-1}(\tau), x_i^{k-1}(\tau))<\pi $. As $ x_i^k(\tau) $ lies in the geodesic path between $ x_{i-1}^{k-1}(\tau) $ and $ x_{i}^{k-1}(\tau) $,  Lemma \ref{lem:convex_cell} dictates that  $ x_i^k(\tau) \in S $. Thus, by induction, we obtain $ x_n^n(\tau) \in S, 0 \le \tau \le 1. $
\end{proof}
 A straightforward conclusion is that the curve $ c_{x_0,x_2,x_4} $ constructed from Algorithm \ref{alg:1} is well-defined and evolves within the cell.

\subsection{B\'ezier curve construction in a set of cells}
Now we apply Algorithm \ref{alg:1} to generate a curve evolving among a set of cells. In the following, we use the notation $c_{x_0,x_2,x_4}(\tau):[0,1]\to SO(3)$ to denote the curve generated from Algorithm \ref{alg:1} given the three points $x_0, x_2, x_4$.

\begin{prop} \label{prop:BezierCurveAll}
Assume that $R_0, R_f $ are the initial and target orientations, respectively, $R_0 \in S_1$, $R_f \in S_m$, and assume there exists a sequence of cells $S_1 S_2 \cdots S_m$ such that $S_i  S_{i+1}$ are adjacent for $i=1, 2, \cdots, m-1$.  Then, a curve $c:\mathbb{R} \supset[0,m] \to SO(3)$ defined as
\begin{equation} \label{eq:curve_all}
c(\tau) = \left\{ \begin{array}{ll}
c_{R_0, R_1, R_{1,2}} ( \tau), & \tau\in [0,1), \\
c_{R_{i-1,i}, R_{i}, R_{i,i+1}} ( \tau -i + 1), & \tau \in [i-1,i), \\
   & \hspace*{-28px} i \in   \{ 2,3, \cdots,m-1\},\\
c_{R_{m-1,m}, R_m, R_f} ( \tau - m +1), & \tau \in [m-1,m],
\end{array}\right.
\end{equation}
where  $R_i $ is the center of cell $ S_i $, $R_{i,i+1} := R_{i} \exp(1/2 \log( R_{i}^\top R_{i+1}))$, has the following properties:
\begin{enumerate}
\item[i.] 	$c(0) = R_0, c(m) = R_f$; 
\item[ii.]  $c(\tau)$ is a $C^2$ curve; 
\item[iii.] $c(\tau) \in \cup_{i = 1}^{m} S_i$ for $\tau \in [0,m]$. 
\end{enumerate}
\end{prop}
\begin{proof}
 Property i can be straightforwardly verified by substituting $ \tau = 1, 2, \cdots,m-1 $ and the fact that $ c_{x_0, x_2,x_4}( 0) = x_0, c_{x_0, x_2,x_4} (1) = x_4 $. As $ c(\tau) $ is a continuous and piecewise smooth curve, we need to check the  left/right velocity/acceleration at $ \tau = 1, 2, \cdots,m-1 $. By differentiating \eqref{eq:curve_all} and using \eqref{eq:BevierVel}, we obtain, for $  i\in \{ 1, 2,3, \cdots,m-1\}$, $ \frac{dc (\tau)}{d\tau} \vert_{\tau= i^{-}} =
 2 R_{i,i+1} \log( R_i^\top R_{i,i+1} ) , \frac{dc (\tau)}{d\tau} \vert_{\tau= i^{+}} =
 2 R_{i,i+1} \log( R_{i,i+1}^\top R_{i+1} ). $
 Note that since $ R_{i,i+1} = R_{i} \exp(1/2 \log( R_{i}^\top R_{i+1})) $, we get $  \log(R_{i}^\top R_{i,i+1}  ) =  1/2 \log( R_{i}^\top R_{i+1}),$
 $  \log( R_{i,i+1}^\top R_{i+1} ) = \log( \exp(-1/2 \log( R_{i}^\top R_{i+1})) R_{i}^\top   R_{i+1}) = 1/2 \log( R_{i}^\top R_{i+1})$. As the left and right derivative coincide at $ \tau = 1, 2, \cdots,m-1 $, $ c(\tau) $ is at least a $ C^{1} $ curve. From \eqref{eq:BezierAcc}, the geometric acceleration at $ \tau = 1, 2, \cdots,m-1 $ satisfies $   \frac{D^2 c (\tau)} {d\tau^2} \vert_{\tau= i^{-}}  =  \frac{D^2 c (\tau)} {d\tau^2} \vert_{\tau= i^{+}} = 0.  $
Thus, $ c(\tau) $ is a $ C^{2} $ curve.

 Property iii can be verified piecewise. Rewrite the curve segments in \eqref{eq:curve_all} in the form of $ c_{R_i^0,R_i,R_i^2}(\tau), \tau\in [0,1], i\in \{1,2,\cdots, m\} $. Proposition \ref{prop:CurveInCell} implies $  c_{R_i^0,R_i,R_i^2}(\tau) \in S_i $. Thus, the concatenation of the curve segments is contained in the union of the cells, which completes the proof.

\end{proof}

\subsection{Time re-parameterization }

 In order to obtain a reference trajectory that solves Problem 1, let $ \tau $ be a smooth function of time, i.e., $ \tau: \mathbb{R}_{\ge}\to [0,m] $  that rescales the trajectory $ c:[0,m] \to SO(3)$ to the time domain $ \gamma: =c  \circ \tau : \mathbb{R}_{\ge} \to SO(3) $.

Numerous smooth transition functions are known. Here we adopt one from \cite{lindemann2009simple} that fits our needs.
\begin{equation} \label{eq:smoothing_function}
s(x) = \left\{ \begin{array}{ll}
0 & x\in ( -\infty,0), \\
 \frac{\rho(x)}{\rho(x) + \rho(1 - x)} & x \in [0,1), \\
1 & x \in [1,\infty)
\end{array}\right.
\end{equation}
with $ \rho(x) := (1/x)e^{-1/x} $. 

%
%

\begin{thm}
 Given a sequence of  cells $ S_1, S_2, \cdots, S_m $ such that $ R_0 \in S_1, R_f \in S_m, S_i S_{i+1} $ are adjacent. Choose $ c:[0,m] \to SO(3)$ defined in \eqref{eq:curve_all} and $ \tau(t) := m s(t/T) $ with $ s(\cdot) $ in \eqref{eq:smoothing_function}. The curve $ \gamma: = c \circ \tau: \mathbb{R}_{\ge} \to SO(3) $  is continuously differentiable, and satisfies $ \gamma (0) = R_0, \gamma (T) = R_f,  d\gamma /dt(0) = d\gamma /dt(T) = 0, D^2\gamma /dt^2 (0) = D^2\gamma /dt^2(T) = 0, \gamma (t) \in \cup_{i \in \mathcal{N}} S_i $ for $ t \ge 0 $.
\end{thm}

\begin{proof}
	Since $ c $ is $ C^2 $ continuous and $ s(\cdot) $ is smooth, $ \gamma $ is also $ C^2 $ continuous. From Proposition \ref{prop:BezierCurveAll} and  the properties of $ s(\cdot) $, it can be checked that  $ \gamma (0) = R_0, \gamma (T) = R_f, D^2\gamma /dt^2 (0) = D^2\gamma /dt^2(T) = 0, \gamma (t) \in \cup_{i \in \mathcal{N}} S_i $ for $ t \ge 0 $. Noticing the fact that $ s(x) $ is a smooth function with $ d^i s/d x^i (0) = d^i s/d x^i (1) =  0 $ for any integer $ i $ (proven in \cite{lindemann2009simple}), we conclude that $  d\gamma /dt(0) = d\gamma /dt(T) = 0$.
\end{proof}
 
 \begin{rem}
 	Although in this work we set the initial and terminal velocities to be zero, the presented method can be directly extended to solve interpolation problems with non-zero velocity boundary conditions by manipulating the controlling points in the cells $ S_0, S_m $ and the time-reparametrization function $ s(\cdot) $. This is a straightforward extension and details are omitted here. 
 \end{rem}
 
  We demonstrate in Fig. \ref{fig:trajectory_illustration} the constructed reference curve $c \circ \tau $ (red line) and the curve from \cite{tan2020constrained} (blue dash line) for comparison. The data is given in the simulation section. It is seen that a smoother attitude maneuver is obtained compared to that of \cite{tan2020constrained}. Figure \ref{fig:traj_comp_vel_magtitude}  further shows that the trajectory in \cite{tan2020constrained}  needs to reach zero velocities at intermediate points, which is avoided in the new construction. The maximal angular velocity magnitude has also decreased compared to that of \cite{tan2020constrained}.

\begin{figure*}[h]
	\centering
	\begin{subfigure}[t]{0.31\linewidth}
	    \includegraphics[width=\linewidth]{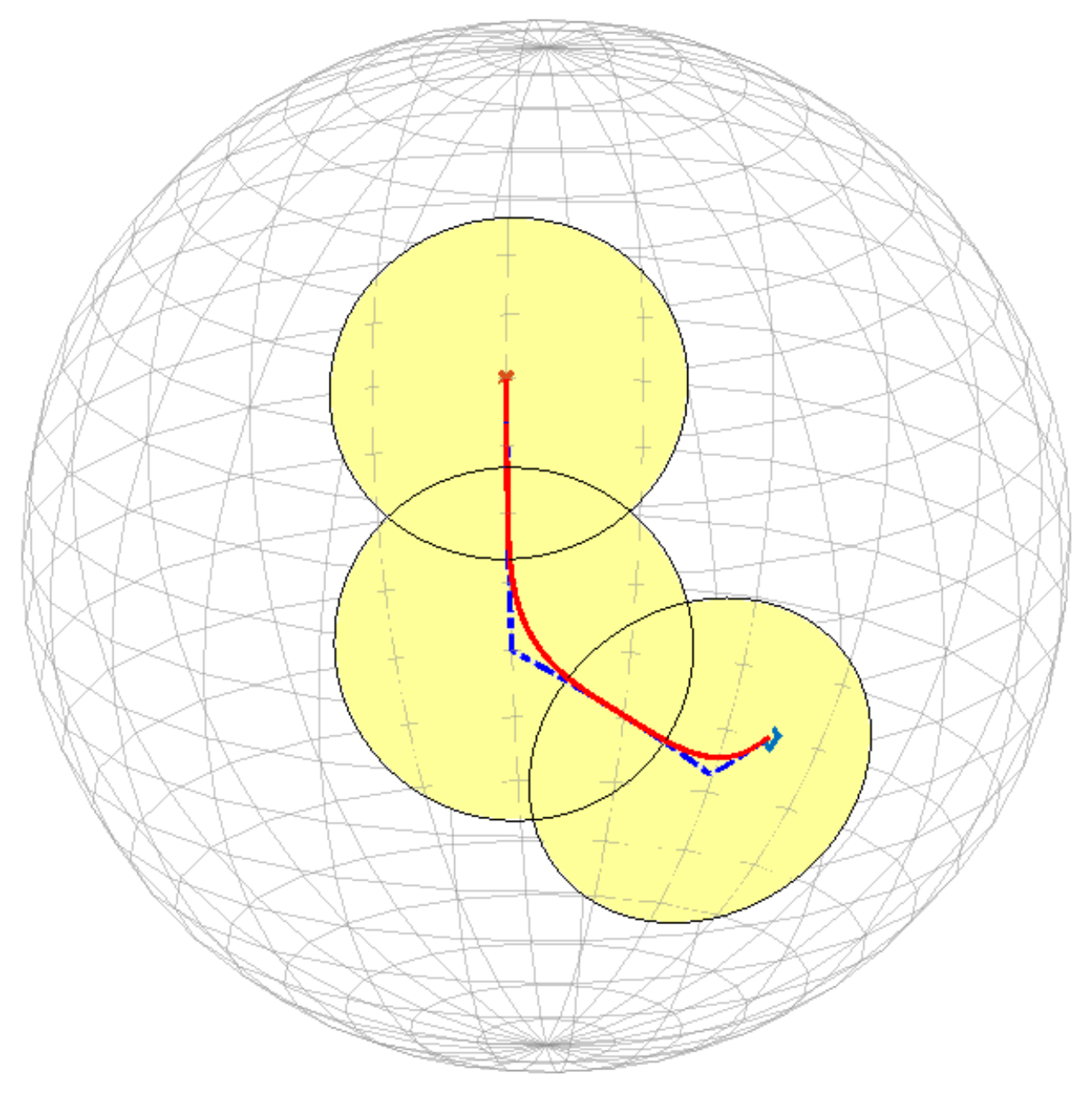}
	 	    \caption{  Trajectory of $x$-axis. }   
	  \end{subfigure}
  \begin{subfigure}[t]{0.31\linewidth}
  	\centering\includegraphics[width=\linewidth]{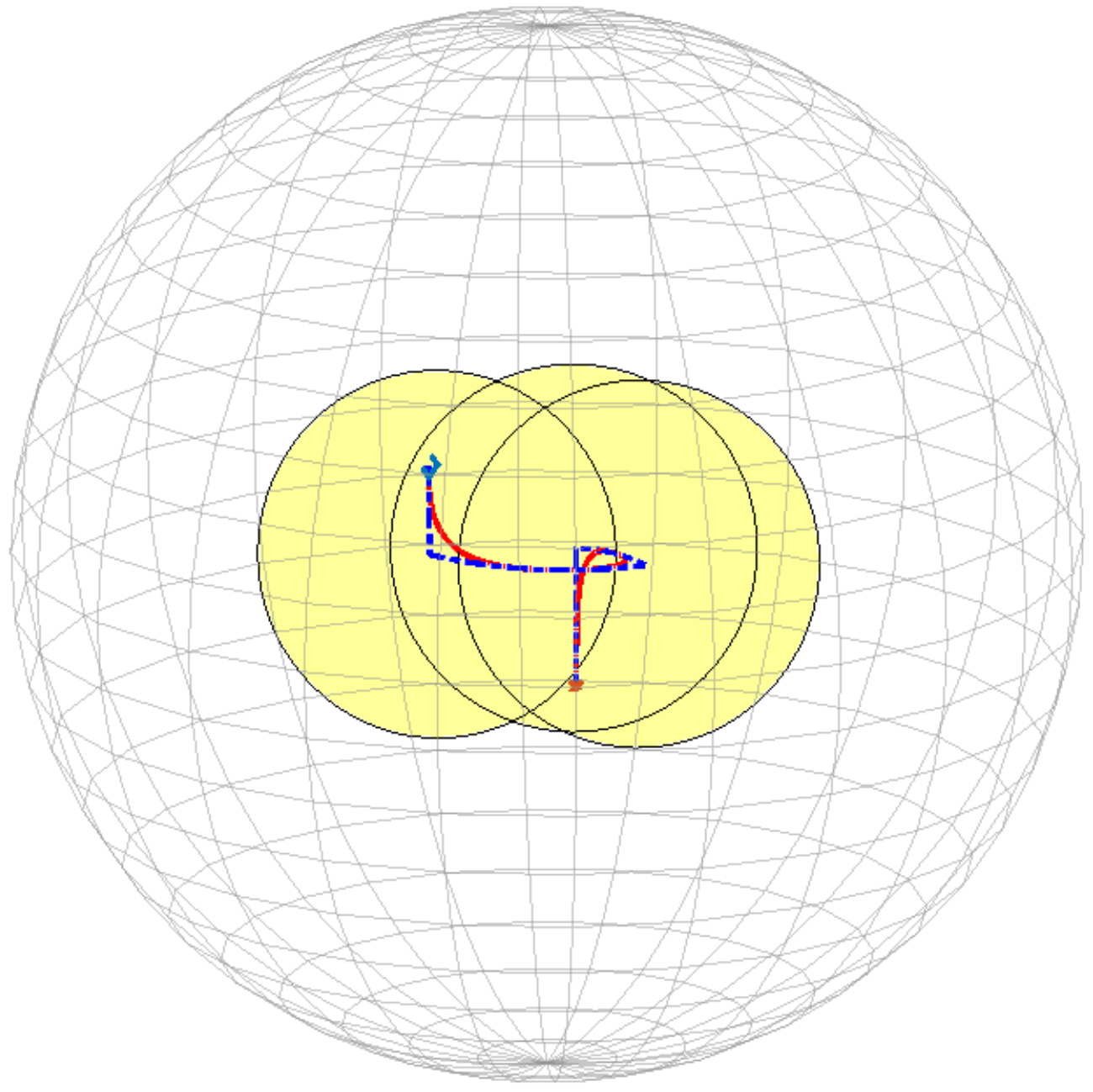}
  	\caption{ Trajectory of $y$-axis.}
  \end{subfigure}
	\begin{subfigure}[t]{0.31\linewidth}
		\centering\includegraphics[width=\linewidth]{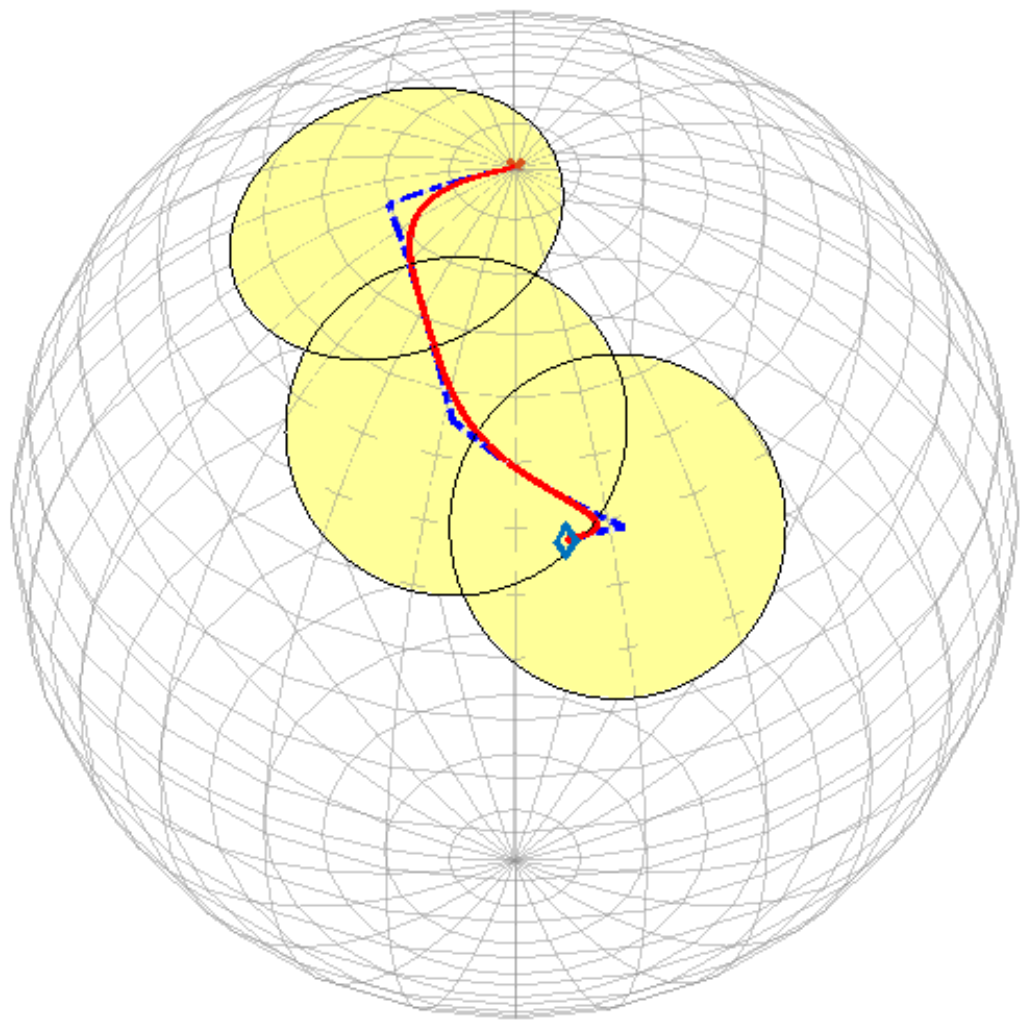}
		\caption{ Trajectory of $z$-axis.}
	\end{subfigure}
	\caption{  Comparison of the trajectories of body-fixed axes: $ c\circ \tau $  in red and the one from \cite{tan2020constrained} in blue. }
	\label{fig:trajectory_illustration}
\end{figure*}

\begin{figure}[h]
	\centering
	\includegraphics[width=\linewidth]{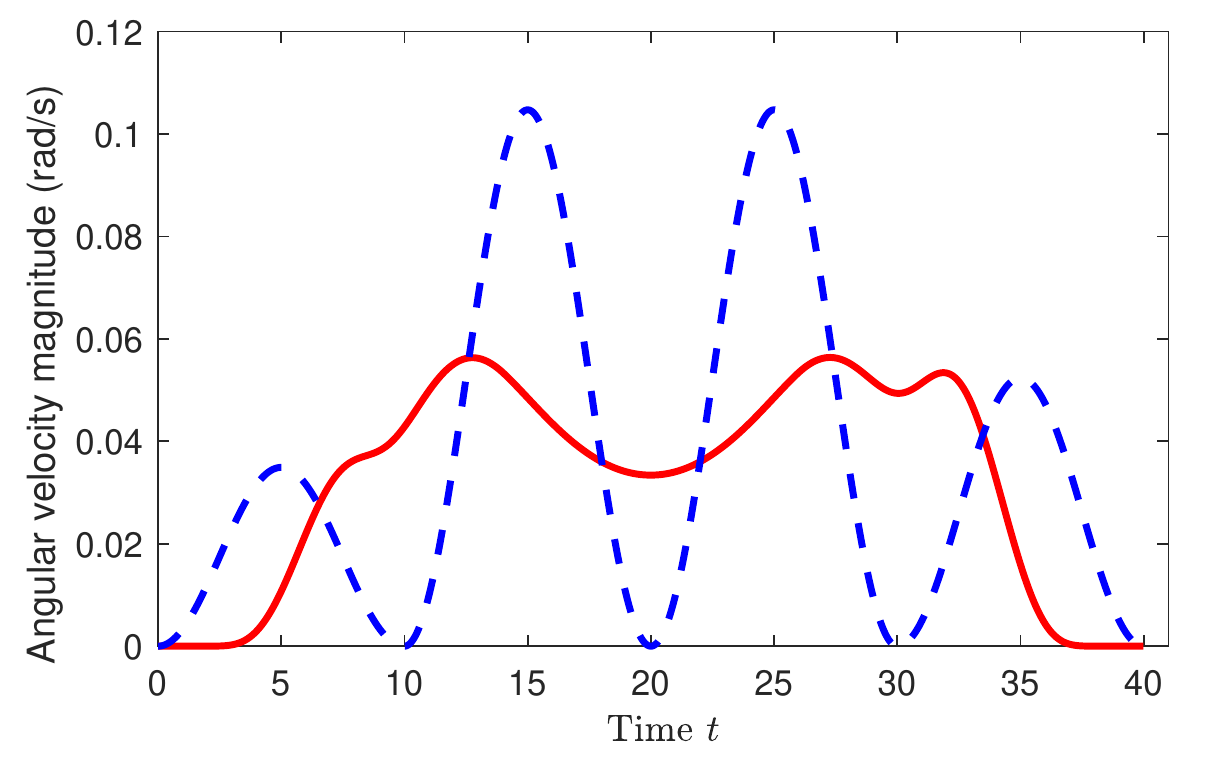}
		\caption{Time histories of the reference angular velocity magnitude of the trajectory $ c\circ \tau $ in red and the one from \cite{tan2020constrained} in blue.}
		\label{fig:traj_comp_vel_magtitude}
\end{figure}

\section{ Control barrier function design } \label{sec:control_barrier}

In this section, we present the procedure to construct a zeroing control barrier function that guarantees the actual attitude trajectory evolves within $ \cup_{i\in \mathcal{N}} S_i $. 

 We start the barrier function design from one cell. For an arbitrary cell  $S_i$, define a function $ r_i: SO(3)\to \mathbb{R} $
\begin{equation} \label{eq:ri}
r_i(R) = \epsilon - \Vert R_i- R \Vert_F^2/2,
\end{equation}
where constant $ \epsilon:= 4 \sin^2(\theta/2)$, $ R_i , \theta $ are the cell center and radius of cell $S_i$, respectively. It is easy to show that $ r_i(R) > 0 $ if and only if $ R \in S_i $, in view of the fact that $ \Vert v- w \Vert_F = 2\sqrt{2}\sin(\text{d}(v,w)/2) $  holds for $ v,w \in SO(3) $. If we need to constrain the trajectory in cell $ S_i $, $r_i(R)$ is a natural candidate as a zeroing control barrier function as it indicates how far the state is from the cell boundary.  Note that there are many alternatives $ r_i(R) $ to  \eqref{eq:ri}, for example, $ r_i(R) = \theta - \text{d}(R,R_i) $. The reason we choose $  r_i(R)$ as in \eqref{eq:ri} is merely to simplify the expression of its derivatives, as shown later.    

 To ensure the actual attitude trajectory evolves within $ \cup_{i\in \mathcal{N}} S_i $, we need for every time instant $ t\ge 0 $, there exists at least one $i\in \mathcal{N} $ that $R(t)\in S_i$, i.e.,
\begin{equation}
\max_{i\in \mathcal{N}} (r_i(R(t))) > 0, \quad \text{for } t\ge 0.
\end{equation}
This  $ \max $ operation would lead to nonsmooth analysis and a complex formulation \cite{glotfelter2017nonsmooth}. In the following, we will show how to circumvent the nonsmooth analysis.

Define
\begin{equation} \label{eq:barrier_func}
h(R) = \sum_{i\in \mathcal{N}} s(r_i(R)/\epsilon) - \delta,
\end{equation}
where $  \delta >0$ is a user-defined constant, and $s(\cdot)$ is given in \eqref{eq:smoothing_function}. The associated constrained set is thus $ C_h^R = \{ R \in SO(3): h(R) \ge 0 \} $. Since $  \cup_{i\in \mathcal{N}} S_i = \{ R: h(R) > -\delta \} $, it is straightforward that $ C_h^R \subset \cup_{i\in \mathcal{N}} S_i $, and the constant $ \delta $ represents the safety margin. The conservativeness is illustrated in Fig. \ref{fig:geometric_intepretation_conservetiveness} in the planar case. For any given $ C^2 $ curve $ c\circ \tau: \mathbb{R}_{\ge} \to \cup_{i \in \mathcal{N}} S_i $, we can find a safety margin (i.e., $ \delta $) such that the curve $ c\circ \tau $ evolves within $ C_h^R $. In the following, we thus assume $  c\circ \tau(t) \in C_h^R $ for $ t\ge 0 $.

\begin{figure}[h]
	\centering
	\includegraphics[width=\linewidth]{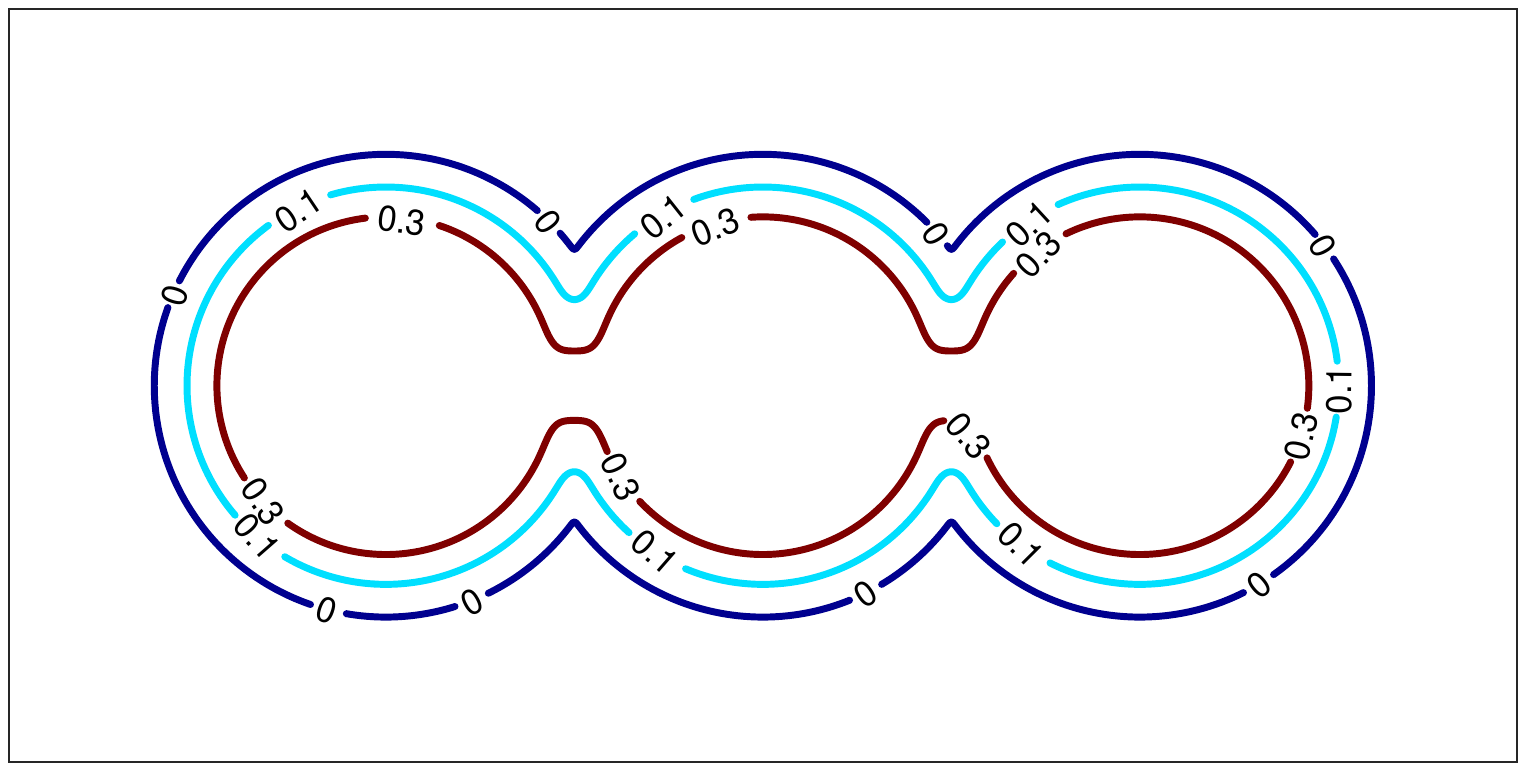}
	\caption{Illustration of $ C_h $ in the planar case with different conservativeness $ \delta^\prime $s.}
	\label{fig:geometric_intepretation_conservetiveness}
\end{figure}

We embed the attitude dynamics in \eqref{eq:dyn} in a higher dimensional Euclidean space  as
\begin{equation} \label{eq:embedded_dyn}
\dot{x} := f(x) + gu,
\end{equation}
where $ x = (r_{11}, r_{12}, \cdots, r_{33}, \omega_1,  \omega_2,  \omega_3 ) \in \mathbb{R}^{12}, f(x) = \big(r_{12}\omega_3 - r_{13}\omega_2; r_{13}\omega_1 - r_{11}\omega_3;r_{11}\omega_2 - r_{12}\omega_1; r_{22}\omega_3 - r_{23}\omega_2;r_{23}\omega_1 - r_{21}\omega_3; r_{21}\omega_2 - r_{22}\omega_1; r_{32}\omega_3 - r_{33}\omega_2; r_{33}\omega_1 - r_{31}\omega_3; r_{31}\omega_2 - r_{32}\omega_1; J^{-1}(-[\omega]_{\times} J \omega)   \big) \in \mathbb{R}^{12}, g =  \begin{pmatrix}
0_{9 \times 3} \\
J^{-1}
\end{pmatrix} $. This is equivalent to \eqref{eq:dyn} by rewriting the attitude dynamics in a vectorized manner. Note that for $ r_{11}, r_{12}, \cdots, r_{33} $, there  exist $ 6 $ implicit equality constraints since they are elements of a rotation matrix. We denote the corresponding $ 6- $dimensional submanifold $ C_{TSO(3)} := \{ x\in \mathbb{R}^{12} : \begin{psmallmatrix}
x_{1} & x_{2} & x_{3} \\
x_{4} & x_{5} & x_{6} \\
x_{7} & x_{8} & x_{9} 
\end{psmallmatrix}  \in SO(3) \}  $. Moreover, if $ x(0) \in C_{TSO(3)} $, then given any control signal $ u $ which is Lipschitz continuous in $ x $, the solution of the dynamical system \eqref{eq:embedded_dyn} satisfies $ x(t) \in C_{TSO(3)} $  for $ t \ge 0 $. This fact can be easily obtained considering that \eqref{eq:embedded_dyn} and \eqref{eq:dyn} are equivalent.  $ h $ in \eqref{eq:barrier_func} is thus a function of the system states $ x $, in particular, of the states $ (x_1,x_2,\cdots,x_9) $. The associate constrained set is $ C_h:=\{ x \in C_{TSO(3)}: h(x)\ge 0 \} $.  
 
For all $ x \in C_h $, we obtain $ L_g h = 0 $, and  $ L_g L_f h(x)  $ may vanish at some points in $ C_h $ (see Appendix for derivations). Here we note that the higher-order control barrier function design developed in \cite{wences2020correct,xiao2019control} is not directly applicable as a result of this issue.  More specifically, to render the set $ C_h $ forward invariant,   the existing methods  enforce it by requiring $ \dot{h}(x) + \alpha(h(x)) \ge 0$ for all $ x\in C_h $. Re-denote  $h_1(x): = \dot{h}(x) + \alpha(h(x)) = L_{f} h + \alpha(h(x)) $. In order to show $ h_1(x) \ge 0$, it is again enforced by a linear constraint  $ L_gL_f h(x) u  + L_f h_1(x) + \beta(h_1(x))\ge 0 $  on $ u $ with $ \alpha,\beta $ extended class $ \mathcal{K} $ functions.  However, this inequality constraint may not be feasible when $ L_g L_f h(x)  $ vanishes.

 A key observation regarding the singularity set  $ \mathcal{D} = \{ x\in C_{h} : L_g L_f h(x) =0 \} $ is given below:
\begin{prop} \label{prop:singular_points}
Let  $  \mathcal{D} = \{ x\in C_{h} : L_g L_f h(x) =0 \} $. Then, there exists a constant $ \xi >0 $ such that $ \inf_{x\in \mathcal{D}} h(x) \ge \xi $.
\end{prop}
\begin{proof}
	See Appendix.
\end{proof}

In the following a remedy is derived to handle the singularity set $\mathcal{D} = \{ x\in  C_{h}:  L_g L_f h(x) = 0  \} $ effectively.  Denote the associated set $ C_{h,\xi} = \{ x\in C_{TSO(3)} : h(x) \ge \xi \}$.  Let $ \chi(\cdot) $ be a twice differentiable function satisfying the following properties
\begin{equation}\label{eq:q}
\left\{ \begin{array}{l}
\chi(0) = 0,\\ \chi(a) = 1, \text{ for } a\ge 1,\\ \frac{d\chi}{d\iota}(a)>0, \text{ for } a < 1.
\end{array} \right.
\end{equation}
Then we smoothly truncate $ h(x) $ by the upper bound $ \xi $, i.e., 
\begin{equation}\label{eq:modified_h}
b(x) = \chi(h(x)/\xi)
\end{equation}
It is easy to verify that $ C_h = C_{b}:=\{x \in C_{TSO(3)}: b(x)\ge 0\} $. Thus, in the following we show the forward invariance of the set $ C_{b} $ instead.

  We adopt the procedure of the  higher-order control barrier function design as in \cite{wences2020correct,xiao2019control}. The idea is briefly presented here: from Brezis version of Nagumo's theorem (see \cite[Theorem~4.7]{Blanchini2015} and \cite[Theorem~4]{redheffer1972theorems} for a detailed account), the forward invariance of the set $ C_b $ is guaranteed by showing that on the boundary of $ C_b $, the system states are directed into the interior or along the boundary of the constrained set. This condition is enforced as $ \dot{b}(x) \ge - \alpha(b(x)) $ for all $ x\in C_b $, where $ \alpha $ is a continuously differentiable, extended class $ \mathcal{K} $ function. Note that $  L_g b = 0 $, then $b_1(x): = \dot{b}(x) + \alpha(b(x)) = L_f b + \alpha(b(x))$ is still a function of the state $ x $. To guarantee the forward invariance of the set $C_{b1}:= \{ x \in C_{TSO(3)}: b_1(x) \ge 0 \}$, using  Brezis version of Nagumo's theorem again, we have the new condition $ \dot{b}_1(x) \ge - \beta(b_1(x)) $ for all $ x\in  C_{b} \cap C_{b1} $, where $ \beta $ is a continuously differentiable, extended class $ \mathcal{K} $ function. Thus, the condition we will enforce in real-time is given as
 \begin{equation}
 	L_g b_1(x) u  + L_f b_1(x) + \beta(b_1(x))\ge 0
 \end{equation}
for all $ x\in  C_{b} \cap C_{b1} $. The feasibility result is as follows:

\begin{prop}
 The inequality condition on $ u \in \mathbb{R}^3 $
\begin{equation}\label{eq:inequality_cond}
L_g b_1(x) u  + L_f b_1(x) + \beta(b_1(x))\ge 0
\end{equation}
is feasible for all $ x\in C_{b} \cap C_{b1}  $.
\end{prop}

\begin{proof}
We examine the condition in two cases. If $ x\in C_{h,\xi} \cap C_{b1} $, we have $ L_g b_1 = 0, L_f b_1 = 0 $, thus inequality in \eqref{eq:inequality_cond} is equivalent to $ \beta(b_1(x)) \ge 0 $, which is trivially satisfied. If $   x\in (C_{h} \setminus C_{h,\xi})  \cap C_{b1}$, $  L_g b_1 \neq 0 $, thus we can always find a $ u $ that satisfies  \eqref{eq:inequality_cond}. 
\end{proof}

Suppose a nominal bounded control input $ u_{nom}(x) $, Lipschitz continuous in $ x $, has been designed for the attitude dynamics and the closed-loop solution tracks the constructed reference trajectory $ \gamma $. We modify the control input online to account for the safety constraints. Concretely, the controller is given by the quadratic program below:
\begin{equation} \label{eq:controllerQP}
\begin{aligned}
& u(x) = \arg   \min_{u \in \mathbb{R}^3}\|u - u_{nom}\|^2 \\
   & \hspace{-30pt} \text{s.t.} \quad L_g b_1(x) u  + L_f b_1(x) + \beta(b_1(x))\ge 0.
\end{aligned}
\end{equation}
 This formulation reflects that the safety constraint has priority over the tracking mission.

\begin{thm} \label{thm:sys_state}
	For the attitude control system in \eqref{eq:dyn}, the controller \eqref{eq:controllerQP} renders the set $ C_b \cap C_{b1}  $ forward invariant.
\end{thm}

\begin{proof}
	The feasibility of the linear inequality constraint on $ u $ is  guaranteed by Proposition \ref{thm:sys_state} for every $ x \in C_b \cap C_{b1}  $. The solution to the quadratic program  \eqref{eq:controllerQP} has a closed-form solution, given by the KKT condition \cite{boyd2004convex}, as
	\begin{equation} \label{eq:controllerSol}
	\begin{aligned}
	u(x) = u_{nom} +   \mu L_g^\top b_1(x)
	\end{aligned}
	\end{equation}
	with 
	\begin{multline*}
	   	\mu  = \left\{ \begin{array}{l}
	0,   \quad  \text{   if }   L_{g}b_1 u_{\text{nom}} + \beta(b_{1}) + L_{f}b_1 \ge 0 , \\
	   \dfrac{- L_{g}b_1 u_{\text{nom}} - \beta(b_1) - L_{f}b_1}{\| L_{g}b_1 \|^2},  \text{ otherwise. }
	\end{array}\right. 
	\end{multline*}

	This is derived from considering whether the constraint in \eqref{eq:controllerQP} is activated or not and thus omitted here. Viewing  $L_g b_1$ in \eqref{eq:lgb1}, the property of $\chi(\cdot)$ and Proposition \ref{prop:singular_points},  we obtain $ L_g b_1 = 0 $ if and only if when $ x\in C_{h,\xi}\cap C_{b1} $, and,  in the meanwhile,  $L_{g}b_1 {u}_{\text{nom}} + \beta(b_1) + L_{f}b_{1} \ge 0 $ is trivially satisfied for $x\in C_{h,\xi}\cap C_{b1} $. Thus the solution is well-defined for every $ x\in  C_b \cap C_{b1}  $. 
	
	
	The solution in \eqref{eq:controllerSol} can be viewed as 
	\begin{equation}
	    u(x) = v_1(x) + v_2(v_3(x))v_4(x)
	\end{equation}
	with $v_1(x) = u_{\text{nom}}(x), v_2(s) = \left\{ \begin{smallmatrix} 
		0, & \text{ if } s \ge 0 \\
		s, & \text{ if } s < 0 \end{smallmatrix}\right. ,  v_3(x) = L_{g}b_1 u_{\text{nom}} + \beta(b_{1}) + L_{f}b_1, v_4(x) =  \frac{- L_{g}^\top b_1 }{\| L_{g}b_1 \|^2}$. For $x \in (C_b \setminus  C_{h,\xi})\cap C_{b1} $, $L_{g}b_1(x) \neq 0$, we obtain $v_1, v_2, v_3, v_4 $ are locally Lipschitz continuous and thus $u(x)$ is locally Lipschitz continuous in $ (C_b \setminus  C_{h,\xi} ) \cap  C_{b1}$. Furthermore, for $x \in  C_{h,\xi}\cap C_{b1} $, we have $u(x) = u_{\text{nom}}(x)$ and thus $u(x)$ is locally Lipschitz continuous in $ C_{h,\xi}\cap C_{b1} $.
	
	 Now we show that the control input $ u(x) $ is continuous at the boundary between $(C_b \setminus  C_{h,\xi})\cap C_{b1}$ and $C_{h,\xi}\cap C_{b1}$.  Assume a Cauchy sequence of points $ \{ x_i \}_{i = 1,2, 3,\cdots} \subset (C_{h} \setminus C_{h,\xi})\cap C_{b1}, \lim_{i\to \infty } x_i = x_0\in \partial  C_{h,\xi} \cap C_{b1}   $. Viewing $ L_f b(x_0) = 0 $,  $ b_1(x_0) = L_f b(x_0) + \alpha(b(x_0))$, we get $ b_1(x_0) = \alpha(b(x_0))  $. Examining closer, it further derives $ b_1(x_0)= \alpha(1)> 0$. We then obtain $ \lim_{i\to \infty }   u(x_i) = u(x_0) $, viewing  the closed-form solution in \eqref{eq:controllerSol} and the facts that $  u_{\text{nom}}(x_i)$ is bounded by definition, $ \lim_{i\to \infty } L_g b_1(x_i) = 0,  \lim_{i\to \infty } L_f b_1(x_i) = 0,   \lim_{i\to \infty }  \beta(b_1(x_i)) = \beta(b_1(x_0)) >0 $. Thus $  u(x) $ is  Lipschitz continuous in $ x$ for all $ C_b\cap C_{b1} $, which guarantees the existence and uniqueness of the system solution.
	 
	 
	 For all $x\in C_b \cap C_{b1}$, we have
	 \begin{equation}
	 	\begin{aligned}
	 		\frac{\partial b }{\partial x} (f(x) + g(x) u ) = b_1&\ge 0; \\
	 			\frac{\partial b_1 }{\partial x} (f(x) + g(x) u ) &\ge 0
	 	\end{aligned}
	 \end{equation}
	 Thus, the vector field lies in the tangent cone of set $C_b \cap C_{b1}$ for all $x\in C_b \cap C_{b1}$. Applying  Brezis version of Nagumo's Theorem \cite[Theorem~4]{redheffer1972theorems} and noticing the locally Lipschitz vector field, we obtain the set $C_b \cap C_{b1}$ is thus forward invariant.
\end{proof}

\begin{rem}
	Compared to the nonsmooth barrier function design in \cite{glotfelter2017nonsmooth}, we formulate a smooth control barrier function and thus avoid the nonsmooth analysis. This formulation comes at the cost of conservativeness in terms of the set difference between $ \cup_{i \in \mathcal{N}} S_i $ and $ C_h^R $. Note that the conservativeness can be explicitly adjusted by choosing a proper $ \delta $. Based on this smooth control barrier function, we restore the solvable optimization problem in \eqref{eq:controllerQP}. This set difference can be viewed as a safety margin in many robotic applications.
\end{rem}

\begin{rem}
	Although \cite{wu2015safety} has studied the application of barrier functions in constrained attitude control problem,  the proposed framework in this paper is generally more advantageous as 1)zeroing instead of reciprocal barrier function is used, which is well-defined even outside of the safety set and is guaranteed to be robust to model perturbations\cite{ames2016control}; 2)here we deal with safety regions of arbitrary shape and the feasibility to the online optimization is guaranteed.
	
\end{rem}

\begin{rem}
	In Theorem \ref{thm:sys_state} we guarantee the forward invariance of the set $ C_b \cap C_{b1} $ instead of $ C_h $. This does not cause conservativeness. For any $ h(x(0)) > 0 $, or equivalently, $ b(x(0))>0 $, there always exists an extended class $ \mathcal{K} $ function $ \alpha(\cdot) $ such that $ b_1(x(0)) = L_f b (x(0)) + \alpha(b(x(0))) >0  $. In this way, $ C_{b1} $ is constructed such that $ x(0) \in C_b \cap C_{b1} $. 
\end{rem}


\section{Simulations}
In this section, we demonstrate the favorable properties of the constructed reference trajectory and the designed zeroing control barrier function. The scenario is given as follows: the inertia matrix of the rigid body is given by $ J = \begin{bsmallmatrix}
5.5   & 0.06 & -0.03\\
0.06 &  5.5   & 0.01\\
-0.03 & 0.01 & 0.1
\end{bsmallmatrix}  \textup{kg}\cdot\textup{m}^2. $ The target attitude is set as $ R_f = I $, and the center points of the sampling cells are given as $ R_3= \exp(15^{\circ}/{180^\circ} \times {\pi}[e_1]_{\times})$, $ R_2= \exp(30^{\circ}/{180^\circ} \times {\pi}[e_2]_{\times})R_3$, $ R_1 = \exp(30^{\circ}/{180^\circ} \times {\pi}[0, 0.447,0.894]_{\times})R_2 $, and the initial attitude $ R_0  =  \exp(10^{\circ}/{180^\circ} \times {\pi}[e_1]_{\times})R_1 $. The radius of the cells is set as $ \theta = 0.3491 $ rad $(20^\circ)  $. The settling time is $ T = 40s $. Based on these data, we show the constructed reference trajectory in red in Fig. \ref{fig:trajectory_illustration}.

 In what follows, we use the saturated controller from \cite{lee2012robust} as the nominal controller in  \eqref{eq:controllerQP}:
\begin{multline} \label{eq:u_norm}
u_{nom} = J \tilde{R}^\top  \dot{\omega}_r + [\tilde{R}^\top \omega_r]_{\times} J \tilde{R}^\top \omega_r \\-k_1(\tilde{R} - \tilde{R}^\top )^{\vee} -k_2 \tanh(\tilde{\omega}),
\end{multline}
where  $ \tilde{R} = R_r^\top R  $,   $ \tilde{\omega}(t) = \omega -\tilde{R}^\top  \omega_r $, $ R_r, \omega_{r} $ are the reference orientation and reference angular velocity obtained from the constructed trajectory $ \gamma $, respectively, $k_1,k_2>0$ are tuning gains and $\tanh(\cdot)$ is the element-wise hyperbolic tangent function. It is shown in \cite{lee2012robust} that the control law \eqref{eq:u_norm} achieves asymptotic convergence of the attitude tracking error from almost all initial conditions.

In the simulations, we augment the control signal in \eqref{eq:u_norm} with an additive signal $ u_{add} = 0.3* \big(\sin(2\pi \frac{t-20}{5}), \sin(\pi \frac{t-20}{5}),-\sin(\pi \frac{t-20}{5}) \big) $ for the time interval $ t\in [20,25] $. This control signal simulates, for example, a human input to the system that could lead to a deviation from the reference trajectory and may even drive the states out of the safe cells. The controller parameters are set as $ k_1 = 0.2, k_2 = 0.2$. The parameters in the control barrier function are chosen as $ \delta = 0.1,\xi = 0.7,\alpha(x) = \beta(x) = x, \chi(x) =\begin{cases}
(x-1)^3+1, & \text{if }x \le 1, \\
1, & \text{if }x > 1.
\end{cases}  $

The simulation results are shown in the following. Fig. \ref{fig:simulated_trajectory_illustration} shows the trajectories in three cases: 1) no additive signal is applied and the control barrier function exists (in blue); 2) additive signal is applied and control barrier function does not exist (in dark red); 3) additive signal is applied and control barrier function exists (in yellow). It is shown that without the additive control signal, the system trajectory is similar to the reference trajectory in Fig. \ref{fig:trajectory_illustration}. However, when the additive signal exists and only the controller in \eqref{eq:u_norm} is applied, the state deviates from the previous trajectory and runs out of the safety cells. Once the control barrier formulation is applied, the resulting trajectory remains in the safety set. This is further supported by the time history of $ b(x) $ in Fig. \ref{fig:traj_bt_history}.

\begin{figure*}[h]
	\centering
	\begin{subfigure}[t]{0.31\linewidth}
		\includegraphics[width=\linewidth]{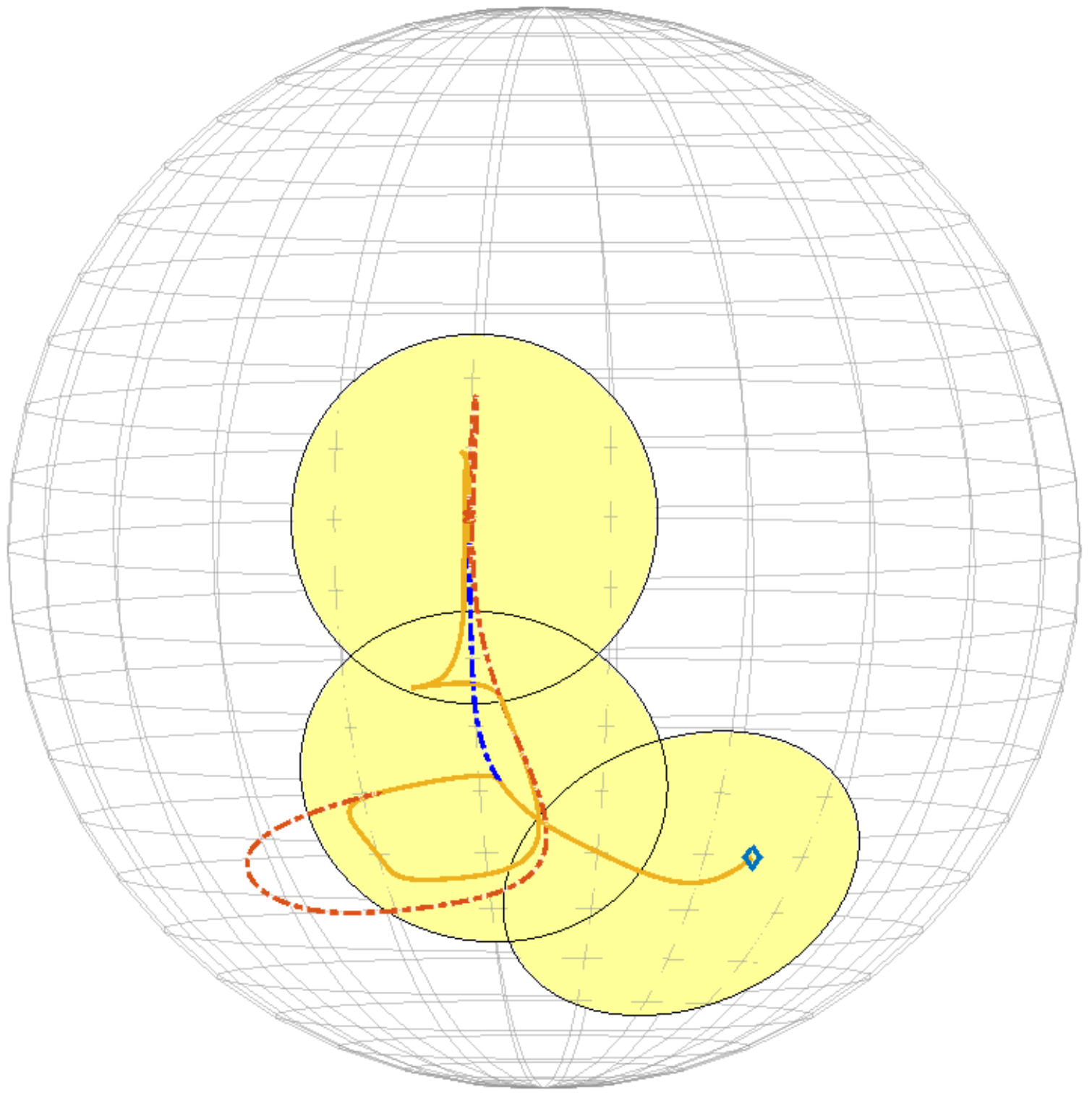}
		\caption{  Trajectory of $x$-axis. }   
	\end{subfigure}
	\begin{subfigure}[t]{0.31\linewidth}
		\centering\includegraphics[width=\linewidth]{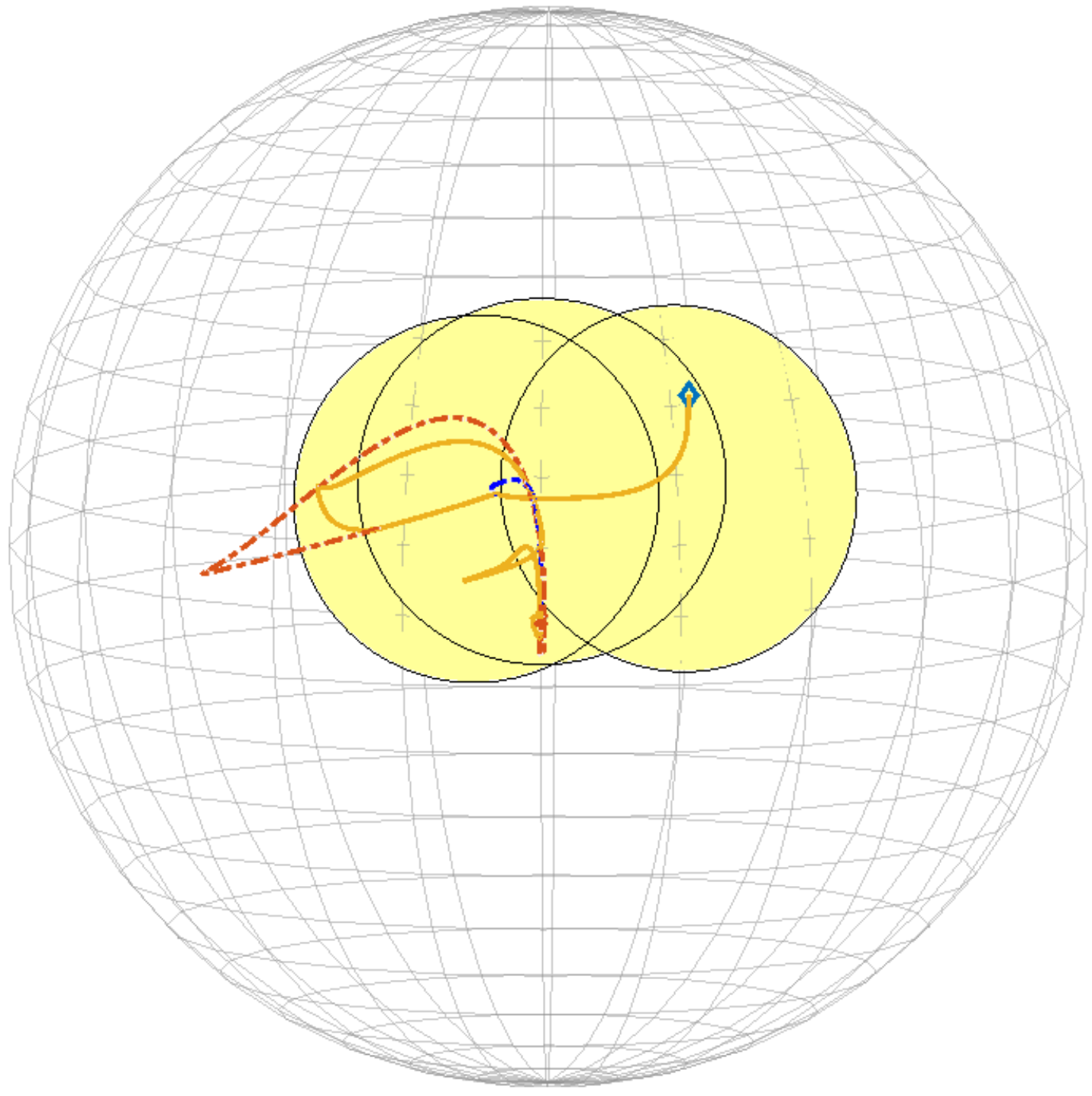}
		\caption{ Trajectory of $y$-axis.}
	\end{subfigure}
	\begin{subfigure}[t]{0.31\linewidth}
		\centering\includegraphics[width=\linewidth]{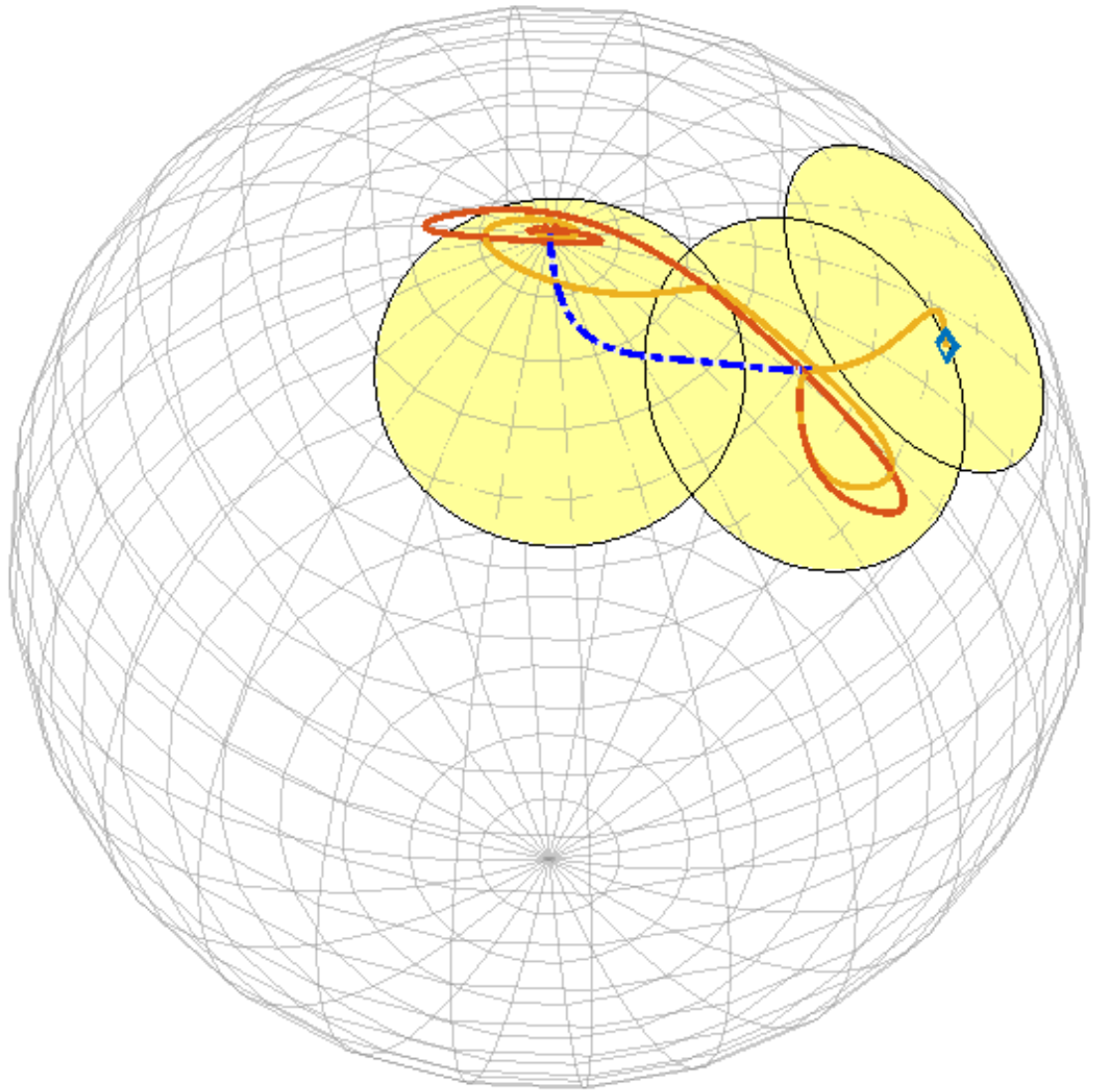}
		\caption{ Trajectory of $z$-axis.}
	\end{subfigure}
	\caption{  Comparison of the trajectories of body-fixed axes in three cases.}
	\label{fig:simulated_trajectory_illustration}	
\end{figure*}

\begin{figure}[ht] 
	\centering
	\includegraphics[width=\linewidth]{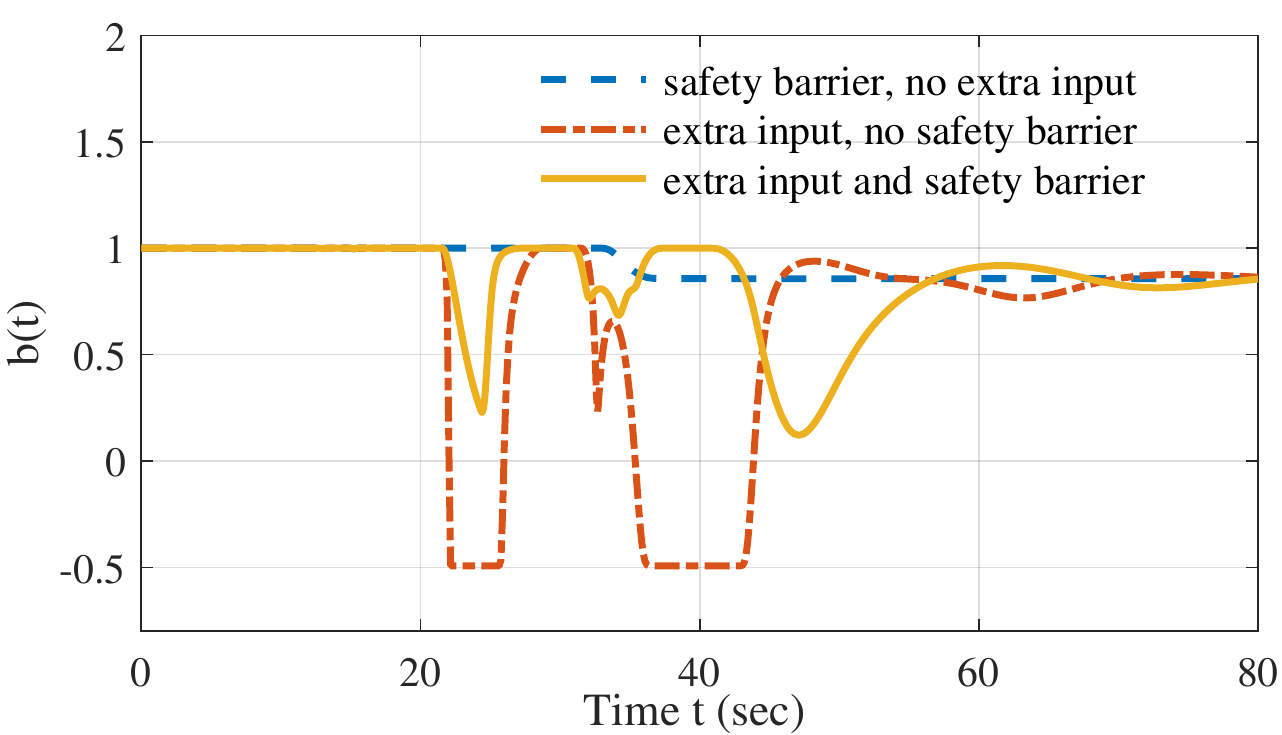}
	\caption{Time histories of the control barrier function $b(x)$ in the three cases.}
	\label{fig:traj_bt_history}
\end{figure}

\section{Conclusion}
In this paper, we construct a $ C^2 $ reference trajectory on $ SO(3) $ and develop a safety certificate utilizing the control barrier function formulation for constrained attitude control problems, following the framework of our previous work in \cite{tan2020constrained}. To construct the reference trajectory, we first design the controlling points for B\'ezier curve generation on $ SO(3) $, which is then time re-parametrized to satisfy boundary conditions. The reference trajectory is shown to be $ C^2 $ continuous, connect the initial and target orientations, and evolve within the predefined safe regions. Moreover, a smooth control barrier function is designed  over a set of overlapping cells to circumvent the non-smooth analysis in previous works. The safety certificate is given as a linear constraint on the control input.  This paper also provides a remedy to handle the states when the singularity of the linear constraint occurs. 

\section*{Appendix}
We collect in this appendix all the results supporting the derivations and claims of the main part of the paper.

From simplicity, denote the auxiliary variables  $ q = (x_1,x_2,\cdots,x_9)$, $ \omega = (\omega_{1},\omega_2,\omega_3) $ and the elements in $ R_i $ as $ \begin{psmallmatrix}
x^i_{1} & x^i_{2} & x^i_{3} \\
x^i_{4} & x^i_{5} & x^i_{6} \\
x^i_{7} & x^i_{8} & x^i_{9} 
\end{psmallmatrix}   $, and let $ [A]_{i,j} $ be the $ (i,j) $th element of matrix $ A $. Thus the state variable is rewritten as $ x = (q, \omega)  $.  As  $ q $ is the stacked vector of the rotation matrix $ R $, we use $ q $ and $R$ interchangeably to denote the attitude state in the following. We obtain  that, for $r_i(q)$ defined in \eqref{eq:ri},

\begin{equation}\label{eq:dridx}
\left[\frac{\partial r_i}{\partial x}\right]_j = \left\{\begin{array}{ll}
-x_j + x^i_j, & j = 1,2,\cdots,9, \\
0, & j = 10,11,12.
\end{array}\right.
\end{equation}
 
 From $ h(q) $ given in \eqref{eq:barrier_func}, we further have $ \frac{\partial h}{\partial x}  = (\frac{\partial h}{\partial q}, \frac{\partial h}{\partial \omega} )  $ with  $  \frac{\partial h}{\partial q} = \sum_{i\in \mathcal{N}} \frac{\partial s(r_i(q)/\epsilon)}{\partial q}  :=  \frac{1}{\epsilon} \sum_{i\in \mathcal{N}} \frac{d s}{d \upeta_i } 	 \frac{\partial (r_i(q))}{\partial q}, \frac{\partial h}{\partial \omega} = \sum_{i\in \mathcal{N}} \frac{\partial s(r_i(q)/\epsilon)}{\partial \omega}  = 0, $
 where $ \upeta_i(q): = r_i(q)/\epsilon $ for brevity. 
 
 
 With  $ f $ in \eqref{eq:embedded_dyn}, we further obtain
 \begin{equation} \label{eq:Lfh}
 \begin{aligned}
 L_f h & = \frac{\partial h}{\partial x} \cdot f = \frac{1}{\epsilon} \sum_{i\in \mathcal{N}} \frac{d s}{d \upeta_i } \frac{\partial (r_i(x))}{\partial x} \cdot f  \\
 & =  \frac{1}{\epsilon}\sum_{i\in \mathcal{N}} \frac{d s}{d \upeta_i } \omega^\top \begin{pmatrix}
 e_{32}^i - e_{23}^i \\
 e_{13}^i - e_{31}^i \\
 e_{21}^i - e_{12}^i 
  \end{pmatrix} := \frac{1}{\epsilon}\sum_{i\in \mathcal{N}} \frac{d s}{d \upeta_i } \omega^\top e^i(q),
  \end{aligned} 
 \end{equation}
 where $ e_{j,k}^i(q) = [R^\top R_i]_{j,k} $ for $ j,k = 1,2,3 $. 
 
 Similarly, we have $   L_g h = \frac{\partial h}{\partial x} \cdot g $. Noticing that $ g $ in \eqref{eq:embedded_dyn} and $ \frac{\partial h}{\partial \omega} = 0 $,  we obtain $  L_g h =  \frac{\partial h}{\partial q} \cdot 0_{9\times 3} + \frac{\partial h}{\partial \omega} \cdot J^{-1} = 0  .$
 
 Moreover, we can calculate that $L_g L_f h = \frac{\partial L_f h}{\partial x} \cdot g = \frac{\partial L_f h}{\partial q}  \cdot 0_{9\times 3} + \frac{\partial L_f h}{\partial \omega} \cdot J^{-1}
 =  \frac{1}{\epsilon}\frac{\partial\sum_{i\in \mathcal{N}} \frac{d s}{d \upeta_i } \omega^\top e^i(q) }{\partial \omega} \cdot J^{-1}$.
  Note that $ \upeta_i(q) $ only relies on $ q $, and thus
  \begin{equation}\label{eq:lglfh}
  \begin{aligned}
  L_g L_f h   =   \frac{1}{\epsilon}\sum_{i\in \mathcal{N}} \frac{d s}{d \upeta_i }  \frac{ \partial\omega^\top e^i(q) }{\partial \omega} \cdot J^{-1}   =   \frac{1}{\epsilon}\sum_{i\in \mathcal{N}} \frac{d s}{d \upeta_i } e^i(q)^\top J^{-1}
  \end{aligned}
  \end{equation}  
 
  \begin{proof}[Proof of Proposition \ref{prop:singular_points}]
  Noting that  $ J $ is positive definite, $ x \in \mathcal{D}:=\{ x\in C_{h}: L_g L_f h = 0 \} $ implies that $ \sum_{i\in \mathcal{N}} \frac{d s}{d \upeta_i } e^i(q)^\top = 0 $. Further noticing that the state can either be in one cell or in the intersection of two cells, we analyze these two cases separately.
  \begin{enumerate}
	\item  If there exists a cell $ S_j, j\in \mathcal{N} $ such that $ r_j(x) >0$, and $ r_k(x) = 0  $ for all $ k \in \mathcal{N},k\neq j $, then $	\sum_{i\in \mathcal{N}} \frac{d s}{d \upeta_i } e^i(q) = 0 \Rightarrow  \frac{d s}{d \upeta_j } e^j(q) = 0 \Rightarrow e^j(q) = 0$
	viewing the property of $ \frac{d s}{d \upeta_i } $. Considering the definition of $ e^i(q) $ in \eqref{eq:Lfh}, we obtain $ R^\top R_i  = I$, i.e., $ R = R_i $, which obviously lies inside $ C_h $. 
	\item If there exist two cells $ S_j, S_k, j,k\in \mathcal{N} $ such that $ R \in S_j\cap S_k$, i.e.,  $ r_j >0, r_k >0 $, then the condition $ 	\sum_{i\in \mathcal{N}} \frac{d s}{d \upeta_i } e^i(q) = 0$  is equivalent to
	\begin{equation}\label{eq:ej_ek}
	\frac{d s}{d \upeta_j } e^j(q) +  \frac{d s}{d \upeta_k } e^k(q)  = 0. 
	\end{equation} 
	As $ x $ lies in the intersection of two cells, we denote $ \exp([v_j]_{\times}) := R^\top R_j,  \exp([v_k]_{\times}) := R^\top R_k  $ for some $ v_j, v_k \in \mathbb{R}^3 $. We can verify that  $ e^j(q) $ is  parallel  to $v_j  $, and  $ e^k(q) $ is  parallel  to  $ v_k $, respectively.  Thus, in order to fulfill condition \eqref{eq:ej_ek}, it suffices that $ v_j \parallel v_k, $
	which means that $ R  $ lies on the geodesic path between $  R_j$ and $ R_k $. Notice again that $ R \in S_j\cap S_K $, from Lemma \ref{lem:cell_to_cell}, we obtain that the system state $ x $ lies in the interior of $ C_h $.
  \end{enumerate}

Thus, the singular point set $ \mathcal{D} $ is composed of all the center points of the sampling cells and certain points on the geodesic path between $ R_j, R_{k}$, where  $S_j$ and $S_k$ are adjacent, $ j,k \in \mathcal{N} $.  It can be checked that there exists a $ \xi>0$ such that $ h(R_i) = s(1)>\xi >0 $ and $ h(R_{\tau}) > \xi >0 $ for $ R_\tau = R_i \exp(\tau \log(R_i^\top R_j)), 0 \le \tau \le 1 $, and thus we obtain  $ \inf_{x \in \mathcal{D}} h(x) \ge \xi.$
\end{proof}

 In the following, we derive the explicit expressions for $ L_f b,  L_f b_1, L_g b_1 $ that are used in the analysis and simulations in previous sections. From \eqref{eq:modified_h}, we obtain $ L_f b = \frac{d\chi}{d\iota} (\frac{h(x)}{\xi}) \frac{\partial h/\xi}{\partial x } \cdot f =\frac{1}{\xi} \frac{d\chi}{d\iota} L_f h$ and $ L_g b  =\frac{1}{\xi} \frac{d\chi}{d\iota} L_g h = 0,$  where $ \iota:= h(x)/\xi $ for brevity.
 
 As $ b_1(x) = L_f b + \alpha(b(x))$, we obtain $ L_f b_1 = \frac{1}{\xi} \left(\frac{1}{\xi} \frac{d^2\chi}{d \iota^2} (L_f h)^2 + \frac{d\chi}{d \iota} L_f^2 h  \right) + \frac{d \alpha}{d b}(b(x))L_f b $ with $L_f^2 h = \frac{1}{\epsilon} \sum_{i\in \mathcal{N}} \left(  \frac{d^2 s}{d \upeta_i^2 } \frac{\partial \upeta_i}{\partial x} \omega^\top e^i(q) +  \frac{d s}{d \upeta_i }  \frac{\partial \omega^\top e^i(q) }{\partial x}  \right) \cdot f = \frac{1}{\epsilon} \sum_{i\in \mathcal{N}} \left( \frac{1}{\epsilon}  \frac{d^2 s}{d \upeta_i^2 } ( \omega^\top e^i(q))^2 +  \frac{d s}{d \upeta_i }  \frac{\partial \omega^\top e^i(q) }{\partial x} \cdot f  \right) $  and 
 \begin{equation}\label{eq:lgb1}
 \begin{aligned}
 L_g b_1 & = L_g L_f b + \frac{d \alpha}{d b}(b(x)) L_g b =  L_g L_f b \\
 & = \frac{1}{\xi} \left( \frac{1}{\xi} \frac{d^2\chi}{d\iota^2}L_g h L_f h  + \frac{d\chi}{d\iota} L_gL_f h \right)  = \frac{1}{\xi}  \frac{d\chi}{d\iota} L_gL_f h
 \end{aligned}
 \end{equation}
 with $ L_gL_f h $ given in \eqref{eq:lglfh}. 
 

\section*{Acknowledgment}
The authors thank Prof. Fatima Silva Leite for the inspiring discussion on the De Casteljau algorithm on $ SO(3) $.

\bibliographystyle{IEEEtran}
\bibliography{IEEEabrv,CDC_fullver}
\end{document}